\documentclass[a4paper,UKenglish,cleveref, autoref, thm-restate]{lipics-v2021}
\usepackage{graphicx}
\usepackage{amsmath}
\usepackage{amsfonts}
\usepackage{amssymb}
\usepackage{xcolor}
\usepackage{algorithm}
\usepackage[noend]{algpseudocode}
\usepackage{lipsum}
\usepackage{soul}
\usepackage{cite}
\usepackage{hyperref}
\usepackage{appendix}
\nolinenumbers
\title{Optimized 2-Approximation of Treewidth}

\author{Mahdi Belbasi}{Pennsylvania State University}{hidden email(s)}{}{}
\author{Martin F\"urer}{Pennsylvania State University}{fhs@psu.edu}{}{}
\author{Medha Kumar}{Pennsylvania State University}{mkumar@psu.edu}{}{}

\authorrunning{M. Belbasi, M. F\"urer, and M. Kumar} 

\Copyright{Author(s)} 



\category{} 

\relatedversion{} 

\supplement{}

\hideLIPIcs  

\EventEditors{NA}
\EventNoEds{0}
\EventLongTitle{Draft}
\EventShortTitle{Draft}
\EventAcronym{Draft}
\EventYear{2021}
\EventDate{2021}
\EventLocation{Draft}
\EventLogo{}
\SeriesVolume{0}
\ArticleNo{0}
\newcommand{\td}{\text{tree decomposition }}
\newcommand{\gtd}{\text{grouped tree decomposition }}
\newcommand{\gtds}{\text{grouped tree decompositions }}

\newcommand{\tw}{\text{treewidth }}

\begin{document}
\maketitle              
\begin{abstract}
 This paper presents a linear FPT algorithm to find a tree decomposition with a 2-approximation of
the treewidth with a significantly smaller exponential dependence on the treewidth. The algorithm runs in time $O(poly(k) 81^k n)$, compared to Korhonen's running time of $O(poly(k) 1782^k n)$ = $O(2^{10.8k} n)$.

\end{abstract}

\section{Introduction}

Since its introduction by Robertson and Seymour~\cite{robertson1984graph, RobertsonS86}, the treewidth -- a measure of how `tree-like' a graph is -- has played a significant role in computer science. When a near-optimal tree decomposition is known, many hard problems can be solved efficiently for graphs with a small treewidth. Unfortunately, computing an optimal tree decomposition and finding the treewidth of a graph are both NP-hard problems~\cite{arnborg1987complexity}. 


Finding an optimal tree decomposition and the exact treewidth was proven to be in FPT~\cite{bodlaender1996linear,korhonen2022improved}. Bodlaender~\cite{bodlaender1996linear} gave the first exact FPT algorithm running in time $2^{\mathcal{O}(k^3)}n$~\cite{bodlaender1996linear} (where $k$ denotes the optimal treewidth, and $n$ is the number of vertices), however the prohibitive running time shifted focus back to approximate solutions. Over time, the primary emphasis shifted from optimizing the dependence on $n$, to reducing the exponential coefficients of $k$~\cite{BelbasiF2022, BelbasiF2021leftmost, Korhonen2021}. Both known linear-time algorithms with single-exponential dependence on $k$ \cite{bodlaender2016c,Korhonen2021}, suffer from very large coefficients of $k$ in the exponent. In this paper, we achieve a significant reduction in the exponential coefficient for an algorithm based on Korhonen's algorithm~\cite{Korhonen2021}.

\subsection{Main Result}

Many ideas of our algorithm are motivated by Korhonen's algorithm \cite{Korhonen2021,korhonen2023singleexponential} which also achieves a $2$-approximation of the treewidth in time linear in $n$, but with a much larger exponential dependence on $k$. 
The algorithm in \cite{Korhonen2021} attains a decrease in the width of a \td $\mathcal{T}$ by \textit{splitting} a root node $r$ with a bag (denoted by $B_r$) of maximum size. To maintain a tree decomposition, the split is propagated through an unbounded number of nodes in the tree decomposition. However with each split, the algorithm in \cite{Korhonen2021} creates $2$ or $3$ copies of all edited nodes, which slows the algorithm. Our algorithm avoids this drawback by relaxing the maximum degree constraint to allow the \textit{merging} of redundant nodes. To avoid any corresponding increase in running time, we specify the structure of the tree decomposition -- a grouped tree decomposition -- which allows us to obtain the following result. We use the notation $\mathcal{O}^*$ to express an upper bound up to a polynomial factor. With this notation, we write the linear dependency on $n$ behind the $\mathcal{O}^*$ notation, i.e., as $\mathcal{O}^*(81^k)n$.
\begin{theorem}\label{thm:thm1}
  Given a graph $G$, a tree decomposition $\mathcal{T}$ of $G$ with width $\leq 4k + 3$ for some integer $k$, in time $\mathcal{O}^*(81^k)n$, one can either construct a tree decomposition of $G$ with width $\leq 2k + 1$, or conclude that $tw(G) > k$.
\end{theorem}

Using Bodlaender's~\cite{bodlaender1996linear} reduction method, as in \cite{Korhonen2021}, one can omit the assumption of a tree decomposition with a width that is at most a $4$-approximation of the treewidth, at a cost of only a factor polynomial in $k$. This gives us the following corollary.

\begin{corollary}
    In time $\mathcal{O}^*(81^k)n$, one can either construct a tree decomposition of $G$ with width $\leq 2k + 1$, or conclude that $tw(G) > k$.
\end{corollary}

While \cite{branchwidth} finds a $2$-approximation of the branchwidth efficiently, the maximum degree constraint for branch decompositions prevents optimal merging of redundant nodes - which we can leverage to speed-up finding a tree decomposition.

\section{Preliminaries}\label{sec:prelims}

A \td $\mathcal{T} = (V_{\mathcal{T}}, E_{\mathcal{T}})$ of a graph $G = (V, E)$, is a tree where each $y \in V_{\mathcal{T}}$ is associated with a set $B_y \subseteq V$, such that $(1)$ for every edge $e \in E$, between vertices $u$ and $v$, there exists $x \in V_{\mathcal{T}}$ with $u, v \in B_x$; $(2)$ all vertices in $V$ appear in some $B_x$; and $(3)$ for all $x, y \in V_{\mathcal{T}}$ with $z \in V_{\mathcal{T}}$ appearing on the unique path between $x$ and $y$, $ B_{x} \cap B_{y} \subseteq B_z$. We refer to $x \in V_{\mathcal{T}}$ as \textit{node} $x$, and $B_x$ as the \textit{bag} of node $x$. The width of $\mathcal{T}$ is one less than the size of the maximum bag $B_x$ of any $x \in V_{\mathcal{T}}$. Assume $\mathcal{T}$ is rooted (arbitrarily). For a node $x \in V_{\mathcal{T}}$, the parent of $x$ is denoted by $p(x)$. 
The subtree $\mathcal{T}_x = (V_{\mathcal{T}_x}, E_{\mathcal{T}_x})$ of $\mathcal{T}$ rooted at a node $x$ is a \td of a subtree of $G$. The graph induced by $\mathcal{T}_x$ is denoted by $G_x = (V_x, E_x)$, where $V_x = \bigcup_{y \in V_{\mathcal{T}_x}} B_y$ and $E_x$ are the edges in $E$ between vertices of $V_x$.


\begin{definition}\label{def:home}
    Given a rooted \td $\mathcal{T}$, a node $x \in V_{\mathcal{T}}$ is the \emph{home} of a vertex $v$ if $v \in B_x$ and $v\not \in B_{p(x)}$. We denote the home of $v$ by $x = h(v)$.
\end{definition}

\begin{definition}
    A \emph{unique home tree decomposition} $\mathcal{T}$ is a rooted tree decomposition where every non-root node $x$ is the home (Def~\ref{def:home}) of exactly one vertex $v \in V$, and the root node $r$ is the home of at least one vertex $v \in V$.
\end{definition}

\begin{lemma}\label{lemma:uhtd}
    Given a \td $\mathcal{T} = (V_{\mathcal{T}}, E_{\mathcal{T}})$ of a graph $G$, with width $k$,  one can construct a unique home tree decomposition $\mathcal{T}' = (V_{\mathcal{T}'}, E_{\mathcal{T}'})$ of $G$ of the same width $k$ in $\mathcal{O}(k^2 |V_{\mathcal{T}}|)$ time.
\end{lemma}
\begin{proof}
    For $\mathcal{T}$ to be transformed into a unique home tree decomposition, we need to ensure that each non-root node in $\mathcal{T}$ is the home of exactly one vertex. We can create the unique home tree decomposition $\mathcal{T}'$ by traversing $\mathcal{T}$ and handling each pair $x,y$ of nodes of $\mathcal{T}$, where $x$ is the parent of $y$,
    as follows:
    \begin{itemize}
        \item If $B_y \subseteq B_x$, we add the child nodes of $y$ as child nodes of $x$ and delete node $y$.
        \item If $B_y \setminus B_x = \{v_1, \dots, v_j\}$ for some $j > 1$, we add a string of forget nodes between $x$ and $y$ such that the bag of each node is exactly $1$ vertex smaller than the bag of its child. Let these forget nodes be $\{y_1, \dots, y_{j-1}\}$. Then, $B_{y_1} = B_{y} \setminus \{v_1\}$, $B_{y_i} = B_{y_{i+1}}\setminus \{v_i\}$ and $B_{y_{j-1}} = B_x \cup B_y$.
        \item Else if $B_y \setminus B_x = \{v_1\}$, then the unique home property is satisfied and we do nothing.
    \end{itemize}
    We can traverse $\mathcal{T}$ in time linear in the number of nodes. Comparing bags of each pair of nodes takes time linear in the size of their bags. At each instance of the second case, we create at most $k$ nodes. 
    Creating $k$ nodes with bags of size $\mathcal{O}(k)$ can be done in time $\mathcal{O}(k^2)$ resulting in an overall time of 
    $\mathcal{O}(k^2|V_{\mathcal{T}}|)$ to construct a unique home tree decomposition $\mathcal{T}'$ from $\mathcal{T}$. 
    \qed
\end{proof}

\begin{definition}
    A \emph{grouped tree decomposition} $\mathcal{T}$ of a connected graph, is a rooted tree decomposition such that:
    \begin{enumerate}
        \item all nodes in $V_{\mathcal{T}}$ are either \emph{main nodes} ($V_{\mathcal{T}}^M$) or \emph{intersection nodes}($V_{\mathcal{T}}^I$),
        \item the root node and all leaves of $\mathcal{T}$ are main nodes, the root node $r$ has a non-empty bag, 
        \item for all $x \in V_{\mathcal{T}}^M \setminus \{r\}$ and $y \in V_{\mathcal{T}}^I$, $p(x) \in V_{\mathcal{T}}^I$ and $p(y) \in V_{\mathcal{T}}^M$,
        \item for all $x, y \in V_{\mathcal{T}}^M \setminus \{r\}$ with $x = p(p(y))$, $|B_y\setminus B_x| = 1$, $B_x \neq \emptyset$, and $B_{p(y)} = B_x \cap B_y$  i.e., $y = h(v)$ for exactly one $v \in V$, and
        \item for all $x, y \in V_{\mathcal{T}}^I$, $B_x \subseteq B_{p(x)}$ and if $p(x) = p(y)$ then $B_x \neq B_y$.
    \end{enumerate}
\end{definition}

For $x \in V_{\mathcal{T}}^M$, $y \in V_{\mathcal{T}}^I$ with $x = p(y)$ are referred to as the \emph{children} or \emph{intersection children} of $x$. Nodes $z$ with $x = p(p(z))$ are referred to as the \emph{grandchildren} of $x$. In a connected graph, $x$ can have at most $2^{|B_x|} - 1$ children. Throughout, we assume a global ordering over $V$.

\begin{lemma}\label{l2}
    Given a \td $\mathcal{T}^* = (V_{\mathcal{T}^*}, E_{\mathcal{T}^*}) $ of a graph $G = (V, E)$ with width $k$ and size $|V_{\mathcal{T}^*}| = O(n)$, one can construct a grouped tree decomposition of $G$ with width $k$, in time $O(k^2 n)$, where $n = |V|$.
\end{lemma}
\begin{proof}
    Using Lemma~\ref{lemma:uhtd}, we first convert $\mathcal{T}^*$ into a unique home tree decomposition $\mathcal{T}$, in time $\mathcal{O}(k^2\, n)$.
    Since the unique home property is satisfied in the tree decomposition $\mathcal{T}$, to construct a grouped tree decomposition $\mathcal{T}'$, we must ensure that the additional structural properties of intersection nodes are satisfied. Given main nodes $x$ and $y$ in $\mathcal{T}$, where $x$ is the parent of $y$, we construct $\mathcal{T}'$ as follows:
    \begin{itemize}
        \item If there exists an intersection child node $z$ of $x$, such that $B_z = B_x \cap B_y$, we add $y$ as a child of $z$.
        \item Else create an intersection child node $z$ of $x$ such that $B_z = B_x \cap B_y$ and $p(y) = x$.
    \end{itemize}
    Searching for an intersection child node $z$ such that $B_z = B_x \cap B_y$ can take time $O(2^k)$. However, since we keep a global order over all vertices in $V$, we might as well construct a temporary binary tree of intermediate nodes -- with each layer representing a vertex (according to the global order) in $B_x$ that is either included (right child) or excluded (left child) in the intersection children below. We add nodes to this binary tree only when we need to add a path to a newly created intersection node below $x$. Using this temporary tree we can check for existing intersection nodes in $O(k)$ time, since $|B_x| = O(k)$ means that the length of the longest path in the tree is at most $O(k)$. Once all children of $x$ have been added below intersection children of $x$, we delete the temporary binary tree and proceed to the next main node in $\mathcal{T}$. 
    As a unique home tree decomposition, $\mathcal{T}$ has at most $n=|V|$ nodes. The time spent on each node $y$ is $\mathcal{O}(k^2)$, as an intersection node with bag size at most $k+1$ is visited, or constructed with an additional $O(k)$ temporary nodes leading to it, resulting in a total time of $\mathcal{O}(k^2\, n)$.
    \qed
\end{proof}

\subsection{Partitions}\label{subsec:ptns}

A crucial step in our algorithm is finding useful partitions of vertex sets $V' \subseteq V$. We allow some parts in a partition to be empty. Partitions of sets $V_1$ and $V_2$ are said to be \textit{compatible} if they agree on the assignment of the elements in $V_1 \cap V_2$. Given $V_1 \subseteq V_2$, the unique partition $P_1$ of $V_1$, compatible with a partition $P_2$ of $V_2$, is said to be the partition \textit{induced} by $P_2$.

\begin{definition}
    A partition $P = (W_1, W_2, W_3, X)$ of $V' \subseteq V$ is a \emph{legal partition} if there are no edges between $W_i$ and $W_j$ for all $i \neq j$.
\end{definition}

If $P_x$ is a legal partition of $B_x$ for some node $x \in V_{\mathcal{T}}$, then there exists at least one legal partition $P^x = (C_1^x, C_2^x, C_3^x, S^x)$ of $V_x$ compatible with $P_x$. $S^x$ is called a separator of $G_x$.
\begin{definition}
For a given tree decomposition $\mathcal{T}$\, and a partition $P^x = (C_1^x, C_2^x, C_3^x, S^x)$ of $V_x$, \emph{the size} of $P^x$ is the size of the separator $S^x$. Furthermore, \emph{the size} of a legal partition $P_x = (W_1^x,W_2^x,W_3^x,X^x)$ of $B_x$, is the minimum size of a compatible legal partition $P^x$ of $V_x$. We denote the size of a legal partition by $\mathrm{size}(P_x)$ or $|P_x|$.
\end{definition}

\begin{definition}\label{def:goodpart}
For a given tree decomposition $\mathcal{T}$, and given nodes $x$ and $y$ such that $x = p(y)$, a legal partition $P_y = (W_1^y,W_2^y,W_3^y,X^y)$ of $B_y$ is \emph{good with respect to a legal partition $P_x = (W_1^x,W_2^x,W_3^x,X^x)$} of $B_x$ if 
\setlength{\parsep}{0.1\parsep}
\setlength{\itemsep}{0.1\itemsep}
\begin{enumerate}
\item $P_y$ is compatible with $P_x$,
\item $\mathrm{size}(P_y)$ is minimal among all legal partitions of $B_y$ compatible with $P_x$, and
\item $|X^y|$ is maximal among the partitions satisfying conditions $(1)$ and $(2)$.
\end{enumerate}
For the root $r$, $P_r$ is \emph{good} if $|W_i^r| + \mathrm{size}(P_r) < |B_r|$ for all $i$ and $\mathrm{size}(P_r) \leq k+1$.

A legal partition $P = (C_1, C_2, C_3, S)$ of $V$ is \emph{good} if $P$ induces a good partition of $B_r$ and for all nodes $x$ and $y$ of $\mathcal{T}$ with $x = p(y)$, the partition of $B_y$ induced by $P$ is good with respect to the partition of $B_x$ induced by $P$.
\end{definition}

Note that at least two sets $W_i^r$ are non-empty for the root $r$. In a good partition $P$ of $V$, the vertices of any connected component of $G[V \setminus S]$ that does not intersect with $B_r$, can be assigned to any part $C_i$. A \emph{split} is a special good partition $P=(C_1, C_2, C_3, S)$ of $V$ that restricts such assignments. For $P$ to be a split, it is not allowed for any node $y$ with partition $P_y$ induced by $P$ to assign $v \in B_y$ to some component $C_i$, if $p(y)$ assigns some vertices to a $C_j$ for $j \neq i$, but none to $C_i$. Further consistencies are enforced if $B_{p(y)} = X^{p(y)}$.

\begin{definition}\label{def:split}
    A \emph{split} of a \td $\mathcal{T}$ is a partition $P = (C_1, C_2, C_3, S)$ of $V$ and a function $a:V_{\mathcal{T}}\rightarrow \mathcal{P}(\{1, 2, 3\})$ with the following properties:
    \begin{itemize}
        \item $P$ is good for $\mathcal{T}$
        \item For every node $x$, and $P_x = (W_1^x, W_2^x, W_3^x, X^x)$ induced by $P$, $P_x$ satisfies:
        \begin{enumerate}[(a)]
            \item $a(x) = \begin{cases}
                \{i \mid W_i^x \neq \emptyset\}, \text{if } \exists i \text{ such that } W_i^x \neq \emptyset \\
                \{ i \} \text{ for one } i \in a(p(x)), \text{otherwise}
            \end{cases}$
            \item $a(x) \subseteq a(p(x))$
        \end{enumerate}
    \end{itemize}
\end{definition}


\begin{lemma}\label{lem:split2k}
    For a graph $G$ with $\text{tw}(G) \leq k$, and a \td $\mathcal{T}$ of $G$ with a root $r$ such that $|B_r| > 2(k+1)$, there exists a split of $\mathcal{T}$.
\end{lemma}

\begin{proof}
Assume the size of the root bag $B_r$ is more than $2(k+1)$.
By Lemma~II.2 of\cite{RobertsonS86}, there is a separator $S$ of $B_r$ of size $|S| \leq k+1$ such that every connected component of $G[V \setminus S]$ contains at most $|B_r|/2$ vertices of $B_r$. These connected components can easily be grouped into 3 sets to produce a legal partition $P = (C_1, C_2, C_3, S)$ of $V$ such that $C_i$ is the set of vertices belonging to the connected components in the $i$th set and $|C_i \cap B_r| \leq |B_r|/2$ for all $i$.
As a consequence, for the induced partition $P_r = (W_1^r, W_2^r, W_3^r, X^r)$
of $B_r$, we have $|W_i^r| + |S| \leq \frac{|B_r|}{2} + k+1 < |B_r|$ for all $i$, implying that $P$ is good at $r$.

Now, we show how to modify $P$ to obtain a split.
Let $\mathcal{T}'=(V_{\mathcal{T}'},E_{\mathcal{T}'})$ be a maximal subtree of $\mathcal{T}$ rooted at $r$ such that there is a pair $(P,a)$ with 
\begin{itemize}
    \item 
    $P$ is a legal partition of $V$, and
    \item
    $P$ and $a$ fulfill properties (a) and (b) of Definition~\ref{def:split} for the nodes in $V_{\mathcal{T}'}$.
\end{itemize}
It is straightforward to see that the root $r$ belongs to $V_{\mathcal{T}'}$.
We claim that $V_{\mathcal{T}'} = V_{\mathcal{T}}$.
Otherwise, there would be a node $y$ with parent $x$ such that (a) and (b) are satisfied in $x$ but not in $y$.

Let $(W_1^y, W_2^y, W_3^y, X^y)$ be the partition of $B_y$ induced by $P$. 
First, assume that (a) is violated in $y$.
Then, for all $i \notin a(x)$, but $W_i^y \neq \emptyset$, $P$ can be modified by moving all vertices of $C_i \cap V_y$ from $C_i$ to a $C_j$ with $j \in a(x)$ and removing $i$ from $a(y)$ to satisfy (a) in $y$.

If (b) is violated in $y$, it just means that there is a better partition $P'$ which differs from $P$ only in $V_y$, is of the same minimal size, and has a larger intersection with $B_y$. If this were the case, the split would choose $P'$ instead (by definition) to satisfy (b).
We have obtained a partition $\mathcal{T}''$ that satisfies (a) and (b) in the vertices of $\mathcal{T}'$ and also in $y$, contradicting the optimality of $\mathcal{T}'$. 
Thus $V_{\mathcal{T}'} = V_{\mathcal{T}}$, and we have obtained a split.
\end{proof}

\begin{definition}\label{def:editable}
    Given a \td $\mathcal{T}$, a node $x \in V_{\mathcal{T}}$ is \emph{editable} with respect to a split $P$ of $V$, if the induced partition $P_x$ of $B_x$ has at least two non-empty components $W_i^x, W_j^x$ ($i \neq j$).
\end{definition}

By the definition of a split, this implies that the parent of an editable node is also editable. We define dynamic programming tables stored at each node $x$, which contains the sizes of all legal partitions of $V_x$. We show that it is straightforward to compute these tables bottom-up, and that once the tables have been computed, a split can be easily selected top-down.


\subsection{Data Structure}\label{subsec:ds}

Given a grouped \td $\mathcal{T}$ of a graph $G$, we use dynamic programming to find a good partition and the corresponding good separator of $G$. We build a data structure of dynamic programming tables (DP-Tables) $A_x$ stored at every node $x \in V_{\mathcal{T}}$, to help us find a split. Each $A_x$ is of size $4^{|B_x|}$ and stores an entry for every possible $4$-partition of $B_x$.

\begin{definition}\label{def:ds}
    Given a partition $P_x = (W_1^x, W_2^x, W_3^x, X^x)$ of $B_x$ for node $x$, the DP-table entry $s(A_x[P_x])$ stores the size of $P_x$. If $P_x$ is not legal then $s(A_x[P_x]) = \bot$.
\end{definition}


\subsection{Overview of our algorithm}\label{sec:oa}
\begin{algorithm}
\caption{2-ApproximationOfTreewidth($G, \mathcal{T}, w, k$)} \label{alg1}
\textbf{Input}: $k\in \mathbb{N}$, $G = (V, E)$, and \td $\mathcal{T}$ of $G$ with width $w \leq 4k+3$ 
\begin{algorithmic}
    \State Construct a \gtd $\mathcal{T}$ from $\mathcal{T'}$ following Lemma~\ref{l2}
    \For{$w = 4k+3$ down to $2k+2$ with step size: $-1$}
        \State Initialize a DFS instance on $\mathcal{T}$
        \While{the DFS instance on $\mathcal{T}$ has not finished}
        \If{A node $x$ with maximum bag size $|B_x| = w+1$ is reached}
        \State MOVE $x$ to the root
        \State Find a Split $(C_1, C_{2}, C_{3}, S)$ of $V$ top-down in $\mathcal{T}$,  
        \If{no split exists} 
        \State \Return FALSE
        \EndIf
        \State Edit all editable nodes $x$, i.e., for all $i \in a(x)$ (Def~\ref{def:split}) replace $x$ with nodes $x_i$
        \Statex \hspace{1.5cm} such that $B_{x_i} = (C_i \cap B_x) \cup S^x)$
        \State MERGE main nodes that aren't the home of a vertex with their grandparent.
        \State Create a new root node $r'$ s.t. $B_{r'} = S$ and the $r_i$ descendants of $r'$
        \State Mark all newly created nodes as unvisited
        \EndIf
        \EndWhile
    \EndFor
    \Return $\mathcal{T}$
\end{algorithmic}
\end{algorithm}
Algorithm~\ref{alg1} shows the pseudo-code of the algorithm which works in rounds. Each round splits all nodes with bags of maximum size $w+1$ (constant for that round) until all nodes have bags of size strictly smaller than $w+1$. When the width of the \td falls to below $2k+1$, a split may not be found, and the $2$-approximation is returned. If no split is found for some width $> 2k+1$, then the \tw is known to be greater than $k$. 
The bottleneck in the algorithm is recomputing the large DP-table stored at each editable node after an operation is performed on the grouped tree decomposition $\mathcal{T}$ with root node $r$.

Like Korhonen\cite{Korhonen2021} we use a potential function $\phi(\mathcal{T})$ to analyze the running time. 
$\phi(\mathcal{T})$ is defined to upper bound the number of DP-Table recomputations per round of Algorithm~\ref{alg1}. Consider the function $\phi(\mathcal{T}) = \alpha(\mathcal{T}) + \beta(\mathcal{T}) + \gamma(\mathcal{T}) + \delta(\mathcal{T})$, where:
    \begin{itemize}
        \item $\alpha(\mathcal{T}) = c_{\alpha} \, \sum\limits_{x \in V_{\mathcal{T}}^M} |B_x|\, f(x)$, where $c_{\alpha} > 0$ and
        $f(x)=
        \begin{cases}
        |B_{x} \setminus B_{p(x)}| \text{ if } x \neq r \\
        (|B_x + 1|)/2 \text{ if } x = r
        \end{cases}
        $, 
        \item $\beta(\mathcal{T}) =  c_{\beta}\,\sum\limits_{x\in V_{\mathcal{T}}^M}$ DFS-Status$(x)$, where $c_{\beta} > 0$ and DFS-Status$(x)$ is its visited status ($2$ if unvisited, $1$ if visited, $0$ if closed) in the DFS instance initialized in Alg~\ref{alg1} 
        \item $\gamma(\mathcal{T}) = c_{\gamma}\, |V_\mathcal{T}|$, where $c_{\gamma} > 0$ 
        \item $\delta(\mathcal{T}) = c_{\delta}\, k^2 |V_{\mathcal{T}}^{\max}|$, where $c_{\delta} > 0$
    \end{itemize}

where, $V_{\mathcal{T}}^M$ denotes the main nodes of $\mathcal{T}$ and $c_{\alpha}, c_{\beta}, c_{\gamma}, c_{\delta}$ are appropriate constants. 
Crucial to the running of the algorithm are the Merge, Split and Move operations. These are described and analyzed in the following sections. In Section~\ref{sec:analysis} we prove that $\phi(\mathcal{T})$ is sufficient to bound the running time, thus completing the analysis of Algorithm~\ref{alg1}.

\section{Details of our Algorithm}

\subsection{Node Operations}

The four main node operations -- Initialization, Addition, Subtraction, and Update -- performed on nodes during Algorithm~\ref{alg1} are described in this section. 
\subsubsection{Initialization}
For a one-node \gtd (or a leaf node) $\mathcal{T} = (\{x\}, \emptyset)$, the DP-table $A_x$ can be Initialized for each legal partition $P_x = (W_1^x, W_2^x, W_3^x, X^x)$ of $B_x$ as $s(A_x[P_x]) = |X^x|$. If $P_x$ is illegal then $s(A_x[P_x]) = \bot$. For a node $x$ with children $y_1, \dots y_l$ in $\mathcal{T}$, we first initialize $x$ in a one-node tree decomposition, and then add $y_i$ to $x$ for all $i \in [l]$. For a \gtd with $O(n)$ nodes and width $w$, an initialization takes $\mathcal{O}^*(4^{w})n$.
\subsubsection{Addition}
Given \gtds $\mathcal{T}_1$ with root $x$, and $\mathcal{T}_2$ with root $y$, where one of $x$ or $y$ is a main node and the other is an intersection node, an Addition creates a new \td where $x$ is the parent of $y$ with $\mathcal{T}_2$ under it. We temporarily allow the root of $\mathcal{T}$ to be an intersection node to describe addition, given the following cases:

\paragraph*{Case 1: Main node $y$ and Intersection node $x$}
Addition is done when $B_y = B_x \cup \{v\}$ for some $v \not \in B_x$. For each $P_x = (W_1^x, W_2^x, W_3^x, X^x)$ of $B_x$, there exist four compatible partitions of $B_y$ corresponding to $v$'s assignment to any one of the four parts. Let $P_y = (W_1^y, W_2^y, W_3^y, X^y)$ be a compatible partition such that $s(A_y[P_y])$ is minimal. We update $A_x[P_x]$ as
$s(A_x[P_x]) = s(A_x[P_x]) + s(A_y[P_y]) - |X^x|$ 
\paragraph*{Case 2: Intersection node $y$ and Main node $x$}
Addition is done when $B_y \subseteq B_x$. We assume that $x$ has no other intersection children with a bag equal to $B_y$. Notice that every partition $P_x= (W_1^x, W_2^x, W_3^x, X^x)$ of $B_x$ can be written as $P_x = (W_1^y \cup W_1^{x \setminus y}, W_2^y \cup W_2^{x \setminus y}, W_3^y \cup W_3^{x \setminus y}, X^y \cup X^{x \setminus y})$, using $P_y = (W_1^y, W_2^y, W_3^y, X^y)$ of $B_y$, and $P_{x\setminus y} = (W_1^{x \setminus y}, W_2^{x \setminus y}, W_3^{x \setminus y}, X^{x \setminus y})$ of $B_x \setminus B_y$. Thus $A_x[P_x]$ is updated using $A_y[P_y]$. If $P_x$ is a legal partition, then $s(A_x[P_x]) = s(A_x[P_x]) + s(A_y[P_y]) - |X^y|$  Clearly an Addition can be done in $\mathcal{O}^*(|A_x|) = \mathcal{O}^*(4^{|B_x|})$.


\subsubsection{Subtraction}\label{subsec:subtr}

Given a \gtd $\mathcal{T}_1$ with root $x$, and a node $y$ such that $x = p(y)$, a subtraction removes $y$ and $\mathcal{T}_y$, from $x$. As in addition, there are two cases:
\paragraph*{Case 1: Main node $y$ and intersection node $x$ }
For each $P_x = (W_1^x, W_2^x, W_3^x, X^x)$ of $B_x$, there exist four compatible partitions of $B_y$. Let $P_y = (W_1^y, W_2^y, W_3^y, X^y)$ be a compatible partition such that $s(A_y[P_y])$ is minimal. We update $s(A_x[P_x]) = s(A_x[P_x]) - s(A_y[P_y]) + |X^x|$ 

\paragraph*{Case 2: Intersection node $y$ and main node $x$}
As in addition, we write partition $P_x$ of $B_x$ as $P_x = (W_1^y \cup W_1^{x \setminus y}, W_2^y \cup W_2^{x \setminus y}, W_3^y \cup W_3^{x \setminus y}, X^y \cup X^{x \setminus y})$, using $P_y = (W_1^y, W_2^y, W_3^y, X^y)$ of $B_y$, and $P_{x\setminus y} = (W_1^{x \setminus y}, W_2^{x \setminus y}, W_3^{x \setminus y}, X^{x \setminus y})$ of $B_x \setminus B_y$. If $P_x$ is a legal partition, then $s(A_x[P_x]) = s(A_x[P_x]) - s(A_y[P_y]) + |X^y|$. Clearly a subtraction can be done in $\mathcal{O}^*(4^{|B_x|})$.



\subsubsection{Update}
Given a \gtd $\mathcal{T}$ with nodes $x, y$ such that $p(y) = x$, an Update of $A_x$ is performed when $A_y$ is modified. This can be done efficiently by subtracting $y$ with its previous table from $x$, and then adding $y$ with the modified table to $x$. Clearly, this takes time $\mathcal{O}^*(4^{w})$ for a \gtd with width $\leq w$. 

\subsection{DFS Instance}\label{subsec:dfs}

Each round of Algorithm~\ref{alg1} proceeds with a single DFS instance over the main nodes in $\mathcal{T}$. Notice that the intersection nodes in $\mathcal{T}$ can at most scale the DFS time by a factor of $2$ and can thus be ignored. We define a function DFS-Status over $V_{\mathcal{T}}^M \rightarrow \{0,1,2\}$ that maps all main nodes that are unvisited, pre-visited or post-visited to $2, 1$ or $0$ respectively. When the round begins, all main nodes are marked unvisited. 
Let $r \in V_{\mathcal{T}}^M$ be the root node of $\mathcal{T}$, that we begin our DFS instance from. As the DFS instance progresses, each main node $x$ is handled as follows:
\begin{itemize}
    \item $x$ is satisfies $|B_x|$ not maximum size: The DFS continues as usual. 
    \item $x$ satisfies $|B_x|$ is of maximum size: The DFS instance is interrupted. $x$ is Moved to the root, and a Split operation is performed. If the Split operation is successful, all new nodes are marked unvisited. The DFS resumes from the first pre-visited node $y$ (i.e., $y$ is not edited during the split) that is adjacent to the unvisited newly created nodes. If no pre-visited node exists, then the DFS resumes from any unvisited node. If the Split operation is unsuccessful, the treewidth of $\mathcal{T}$ must be greater than $k$.
\end{itemize}
The decrease in the $\alpha$-potential caused by the Split pays for the added $\beta$-potential of the newly created nodes.

\begin{lemma}\label{lem-DFScorrectness}
    All nodes are marked post-visited when the DFS-instance stops.
\end{lemma}
\begin{proof}
    Clearly, if a node $x$ is not split, the DFS continues as usual. When the DFS reaches a node $x$ that is to be Split, there must be a chain of pre-visited nodes between $x$ and $r$, the previous root. When $x$ is moved to the root, this path of pre-visited nodes forms a path of descendants of $x$. All nodes edited during the split, form a connected sub-tree below the new root $x$. Due to the structure of grouped tree decompositions, no non-editable intersection node can have an editable main node child. If the old root $r$ lies within the editable sub tree, the DFS resumes from some descendant main node $y$ of $r$ that is not edited and is marked pre-visited. Else there must exist a main node $y$ on the path between $r$ and $x$ that is pre-visited and not edited. In either case, when the DFS resumes from $y$, all new unvisited nodes are pre- and post-visited first, before $y$ is post-visited. The DFS beyond $y$ resumes as usual. Thus, all nodes in $\mathcal{T}$ will be post-visited when the DFS instance stops.
\end{proof}

\subsection{Move}\label{subsec:Move}

\begin{algorithm}[ht!]
    \caption{Move($\mathcal{T}, r, x, P, A_{\mathcal{T}}$)}\label{algMove}
    \textbf{Note: } $r, x \in V_{\mathcal{T}}^M$ and $P$ is the path from the root $r$ to $x$.
    \begin{algorithmic}
        \For{$y$ in $P$ from below $r$ to $x$}
            \State $\mathcal{T}$ $\leftarrow$ Rotate($\mathcal{T}, r, y, A_{\mathcal{T}}$)
            \State $r \leftarrow y$
        \EndFor
        \State \Return $\mathcal{T}$
    \end{algorithmic}
\end{algorithm}
s
Once a maximum-size bag has been split, the DFS instance on $\mathcal{T}$ resumes until it finds a node $x$ such that $|B_x| = w+1$. The node $x$ must then be moved to the root of $\mathcal{T}$, and the DP-tables of all nodes on the path $P$ between $r$ and $x$ in $\mathcal{T}$ updated to reflect the new root. This Move, described in Alg~\ref{algMove}, is performed as a series of edge rotations (Alg~\ref{algRotate}) between consecutive nodes $y_i, y_{i+1}$ in $P$.
\begin{algorithm}[ht!]
    \caption{Rotate($\mathcal{T}, r, y, A_{\mathcal{T}}$)}\label{algRotate}
    \textbf{Note: } $r, y \in V_{\mathcal{T}}^M$ and $r = p(p(y))$.
    \begin{algorithmic}
        \State Let $z = p(y)$
        \State Subtract $y$ from $z$ and Update the table $A_r$
        \If{there are no nodes $y'$ such that $z = p(y')$}
            \State Subtract $z$ from $r$ 
        \EndIf
        \If{$|B_r \setminus B_{z_y}| = j > 1$}
            \If{$\not \exists z_y \in V_{\mathcal{T}}^I$ such that $y = p(z_y)$ AND $B_{z_y} = B_y \cap B_r$}
                \State Initialize $z_y$ such that $y = p(z_y)$ and $z_y \in V_{\mathcal{T}}^I$ and $B_{z_y} = B_y \cap B_r$
            \EndIf
            \State Add $j-1$ main nodes $r_1, \dots r_{j-1}$ between $r$ and $z_y$ such that $|B_{r_1}\setminus B_{r}| = 1$, $|B_{r_{i+1}}\setminus B_{r_{i}}| = 1$ for $i \in \{2, \dots, j-2\}$, and $|B_{z_y}\setminus B_{r_{j-1}}| = 1$
            \State Add corresponding intersection nodes $z_{r_i}$ for $i \in [j-1]$ for each node $r_1, \dots r_{j-1}$
            \State Add $r$ to $z_{r_1}$
            \State Update all new nodes up to $z_y$
        \ElsIf{$|B_r \setminus B_{z_y}|$ = 1}
            \State Add $r$ to $z_y$.
        \EndIf
        \State Update $A_y$ to reflect the change in $z_y$. 
        \If{$|B_r \setminus B_y| = 0$}
            \State Merge $r$ into $y$
        \EndIf
        \State \Return $\mathcal{T}$
    \end{algorithmic}
\end{algorithm}
Notice that whenever the bag of the main node is $j > 1$ vertices larger than its new intersection node parent, up to $O(k)$ nodes may be created. However these nodes have bags with sizes decreasing by $1$, which decreases the corresponding DP-table size by a factor of $4$. The cost of creating and updating the tables of these new nodes is thus, at most $4/3$ times the cost of updating the table of the new root $r$. Correspondingly, when an edge is rotated, and a main node is no longer home to any vertices, it must be \textit{merged} into its grandparent. Thus each rotation takes time within a constant factor of the time taken to update a single DP-table, i.e. $4^{|B_r|}$. Furthermore, the number of rotations performed during a round is bounded by the number of steps taken by the DFS-instance (Lemma~\ref{lem:boundRotations}).



\begin{lemma}\label{lem-move-nochange}
    Given a \gtd $\mathcal{T}$, $\alpha(\mathcal{T})$ does not change while Moving a node with a maximum size bag to the root.
\end{lemma}

\begin{proof}
    Consider a \gtd $\mathcal{T}$ rooted at a node $r$. Let $\mathcal{T}^M$ be the corresponding unique home tree decomposition of $\mathcal{T}$.
    To make the analysis simpler, the Move operation is analyzed over $\mathcal{T}^M$. Notice that the $\alpha$ potential is defined over all main nodes of $\mathcal{T}$, and analyzing it over $\mathcal{T}^M$ is sufficient.

    Assume that $B_r$ is the home of $\ell + 1$ vertices $v_1, \cdots, v_{\ell + 1}$. 
    For a simpler proof, let us assume that the existing $\alpha$ potential of the root node 
    ($\frac{|B_r|(|B_r|+ 1)}{2} = \sum_{j=1}^{|B_r|} j$) is distributed over a chain of virtual (main) nodes $r_1, \dots, r_{\ell}$ above $r$, such that $r_{\ell}$ is the parent of $r$ with $B_{r_{\ell}} = B_r\setminus\{v_{\ell + 1}\}$. Also, for all $i \in \{1, \dots, \ell-1\}$, $B_{r_i} = B_{r_{i+1}}\setminus\{v_{i+1}\}$. We refer to this constructed tree with virtual nodes as $\mathcal{T}^M_V$. Notice that by adding the virtual nodes, we ensure that the unique home property is maintained at \textit{all} nodes in $\mathcal{T}^M_V$, including the root node $r$.
    
    The Move operation is performed on $\mathcal{T}^M$ to change the current root $r$ to some new node $r'$ with a bag of maximum size. This operation is done step-wise by rotating the edges on the path $P$ between $r$ and $r'$. During the move, the chain of virtual nodes $\{r_1, \dots, r_{\ell}\}$ in $\mathcal{T}^M_V$ is maintained with respect to the current (temporary) root $r''$ of the tree, i.e., the variable $\ell$ changes according to the number of vertices with their home in $r''$, namely $|B_{r''}|$.
    
    Given nodes $x, y$ that occur successively on $P$ such that $x = p(y)$ in $\mathcal{T}^M$, when the edge between $x$ and $y$ is rotated, one of the following must occur:
    \begin{itemize}
        \item $|B_x| < |B_y|$: Since $\mathcal{T}^M$ is a unique home tree decomposition, $B_x \subset B_y$ and $|B_y \setminus B_x| = 1$. To maintain the unique home property, node $x$ must be merged into node $y$. Since the new root $y$ is now home to one additional vertex, $\ell$ grows to $\ell + 1$, and an additional $r_{\ell+1}$ node must be added to the chain of virtual nodes above the root in $\mathcal{T}^M_V$. The $\alpha$ potential given to node $r_{\ell+1}$ is the same as the $\alpha$ potential of $x$. Thus, after the rotation, the $\alpha$ potential of $\mathcal{T}^M_V$ -- and correspondingly, $\mathcal{T}^M$ -- remains unchanged. 
        \item $|B_x| > |B_y|$: Let $|B_x \setminus B_y| = t$ for $t > 1$, $B_y = \{u_0, u_1 \dots, u_s\}$, and $B_x = \{u_1, \dots, u_{s+t}\}$. After the rotation, $x$ is the home of vertices $u_{s+1},\dots,u_{s+t}$. A chain of main nodes, namely $x_{1}, \dots, x_{t-1}$ are added between $y$ and $x$ in $\mathcal{T}^M$, where $B_{x_{t}} = B_{x}\setminus\{u_{s+t}\}$ and $B_{x_{i}} = B_{x_{(i+1)}}\setminus\{u_{(s+i)}\}$ for $i\in\{1,\dots, t-1\}$. Notice that since $|B_y| < |B_x|$ and $y$ is home to fewer vertices, the chain of virtual nodes above $y$ in the new $\mathcal{T}^M_V$ is exactly $t$ nodes shorter than the chain of virtual nodes above $x$ in $\mathcal{T}^M_V$, when $x$ was the root. The $\alpha$ potential of the $t$ virtual nodes $r_{1}, \dots, r_{t-1}$ that are no longer required after $y$ is made the root node of $\mathcal{T}_M$, is exactly the same as the $\alpha$ potential of main nodes $x_{1}, \dots, x_{t-1}$, that are added between root $y$ and node $x$ in $\mathcal{T}_M$ (and correspondingly, $\mathcal{T}^M_V$), to maintain the unique home property. Thus, the $\alpha$ potential of $\mathcal{T}^M$ remains unchanged. 
        \item $|B_x| = |B_y|$: $\mathcal{T}^M$ is a unique home tree decomposition and so $|B_y \setminus B_x| = 1$. Since $|B_x| = |B_y|$ are equal, it follows that $|B_x \setminus B_y| = 1$, and nothing more must be changed in $\mathcal{T}^M$ (or $\mathcal{T}^M_V$), i.e., the $\alpha$ potential remains unchanged.
    \end{itemize}
\end{proof}

\begin{lemma}\label{lem:betadecrease}
    The $\beta$-potential decreases by $1$ per step in the DFS-Instance. 
\end{lemma}

\begin{proof}
    In each step of the DFS-Instance, the DFS moves to a different node. The label of the node it moves to then changes from unvisited to pre-visited, or pre-visited to post-visited. Both these changes mark a decrease in the DFS-Status of the node it moves to -- thus decreasing the $\beta$-potential of $\mathcal{T}$.  
\end{proof}

\begin{lemma}\label{lem:boundRotations}
    The number of rotations performed in a round is bounded by the number of steps in the DFS-Instance.
\end{lemma}
\begin{proof}
    All nodes with maximum-size bags are found in a round using the DFS-Instance. The maximum number of rotations that can occur are then, at most, a constant factor times the number of steps in the DFS-Instance.
\end{proof}

\subsection{Merge}

The unique home property is essential for Algorithm~\ref{alg1}. To maintain this property, \emph{main nodes} are merged if they aren't the home of some vertex, i.e., given $x, y \in V_{\mathcal{T}}^M$ such that $x = p(p(y))$ and $B_y \subseteq B_x$, $y$ is merged into $x$. A Merge operation is performed bottom-up as described in Algorithm~\ref{algMerge}.

\begin{algorithm}[hbt!]
    \caption{Merge($\mathcal{T}$, $x$, $y$, $A_{\mathcal{T}}$)}\label{algMerge}
    \begin{algorithmic}
        \State Let $I_x = \{ i \mid i \in V_{\mathcal{T}}^I , x = p(i) \}$ and $I_y = \{i \mid i \in V_{\mathcal{T}}^I , y = p(i) \}$
        \State Let node $z \in I_x$ such that $z = p(y)$
        \If{$\exists y' \neq y$ such that $z = p(y')$}
            \State Subtract $y$ from $z$
        \Else
            \State Delete $z$ from $I_x$
        \EndIf
        \For{$z_y \in I_y$}
            \If{$\exists z_x \in I_x$ such that $B_{z_x} = B_{z_y}$}
                \For{Partition $P = (W_1, W_2, W_3, X)$ of $B_{z_x}$ and $B_{z_y}$}
                    \If{$P$ not legal}
                        \State $s(A_{z_x}[P]) = \bot$
                    \Else
                        \State $s(A_{z_x}[P]) = s(A_{z_x}[P]) + s(A_{z_y}[P]) - |X|$
                    \EndIf
                \EndFor
            \Else
                \State Add $z_y$ as a child of $x$
            \EndIf
        \EndFor
        \State Delete $y$ and any intersection nodes of $y$
        \State Return $\mathcal{T}$
    \end{algorithmic}
\end{algorithm}

The running time of Alg~\ref{algMerge} is determined by the number of complete DP-table updates required. Notice that a full DP-table is only updated when intersection nodes overlap between $x$ and $y$, i.e. $I_x \cap I_y$. Consider a function $\gamma(\mathcal{T}) = c_{\gamma}|V_{\mathcal{T}}|$, for some $c_{\gamma} > 0$. The number of DP-table recomputations performed during the Merge can be counted using $\gamma(\mathcal{T})$.

\begin{lemma}\label{lem-merge-dec}
    Given a \gtd $\mathcal{T}$ that is transformed into a \gtd $\mathcal{T}'$ by a Merge, $\gamma(\mathcal{T}) - \gamma(\mathcal{T}')$ bounds the number of DP-tables recomputed during the Merge.
\end{lemma}
\begin{proof}
    A main node $y$ is merged into its grandparent $x$ if $B_y \subseteq B_x$. As described in Alg~\ref{algMerge}, a DP-table is recomputed at every node $i \in I_x$ such that there exists a node $j \in I_y$ (where $I_z$ denotes the intersection node children of a node $z$) that satisfies $B_i = B_j$. 
    Observe that a node is deleted for each update during the merge, i.e., $j \in I_y$ for each $A_i$ updated, and $y$ for $A_x$. The deleted nodes decrease the $\gamma$ function. Thus $\gamma(\mathcal{T}) - \gamma(\mathcal{T}')$ bounds the number of complete DP-table updates that occur during a Merge.
\end{proof}

\subsection{Split}

A Split operation involves finding a \emph{split} (Def~\ref{def:split}) of $\mathcal{T}$ where the root $r$ has a bag of maximum size $w+1$, and using it to replace $r$ with nodes of smaller bag size. We describe the procedure in Algorithm~\ref{algSplit}, but the operation comprises the following three stages:
\begin{enumerate}
    \item Finding a split $P = (C_1, C_2, C_3, S)$ of $\mathcal{T}$: This is done top-down in $\mathcal{T}$.
    \item Using $P$ to create up to three copies ($x_1, x_2, x_3$) of every editable node (Def~\ref{def:editable}) $x$ in $\mathcal{T}$, with each $B_{x_i}$ containing only vertices in $C_i \cup S$.
    \item Preserving the unique home property of \gtds and creating a new root node $r'$ such that $B_{r'} = S$.
\end{enumerate}

It is important to note that the tree traversal stops further exploration when it reaches a (intersection) node that is not editable, i.e., it only explores the editable tree. This is crucial for maintaining the linear running time. Yet again, the running time of the Split operation is bounded by the number of table updates performed, which we bound using the potential functions $\alpha(\mathcal{T})$ and $\delta(\mathcal{T})$.


\begin{algorithm}[ht!]
    \caption{UpdateTables($x$, $x_i$, $P^g$, $\mathcal{T}$, $\mathcal{T}_i$, $S$)}\label{algUpdTbl}
    \begin{algorithmic}
        \If{$x$ is an intersection node}
            \For{$y \in $ editable children of $x_i$}
                \If{there exists $v \in B_y, v \notin B_{x_i}, v \in S$}
                    \State Create node $x_i^y$ such that $B_{x_i^y} = B_{x_i}$ and set $p(x_i^y) = p(x_i)$
                    \For{$v \in B_y, v \notin B_{x_i}, v \in S$}
                        \State Set $B_{x_i^y} = B_{x_i}^y \cup \{v\}$
                    \EndFor
                    \State Add $y$ to $x_i^y$
                \Else
                    \State Add $y$ to $x_i$
                \EndIf
            \EndFor
            \If{$\exists\, x_i'$ s.t. $B_{x_i'} = B_{x_i}$ and $p(x_i') = p(x_i)$, and $A_{x_i}$ has been updated}
                \State Add $x_i$ to $x_i'$, move all children of $x_i$ to $x_i'$, delete $x_i$
            \EndIf
            \State \Return  
        \EndIf
        \For{$y \in $ editable children of $x$}
            \State Subtract $y$ from $x$
        \EndFor
        \For{ $i \in a(x)$ and every $P_i = (W_1^i, W_2^i, W_3^i, X^i)$ of $B_{x_i}$}
            \State Let $Q = B_x \cap (W_j \cup W_k)$ for chosen good partition $P^g = (W_1, W_2, W_3, X)$ of $B_x$
            \State Let $P$ be the partition $(W_1^i, W_2^i, W_3^i, X^i\cup Q)$ of $B_x$
            \State $s(A_{x_i}[P_i]) = s(A_x[P]) - |Q|$
        \EndFor
        \For{$y \in $ editable children of $x_i$}
            \State Add $y$ to $x_i$
            \State Let $X^y = S \, \cap B_y $ be the separator of $\mathcal{T}_y$ at $y$
            \For{$v \in X^y, v \not \in B_{x_i}$}
                \State Add $v$ to $B_{x_i}$
                \For{$P = (W_1^i, W_2^i, W_3^i, X^i)$ of $B_{x_i}\setminus \{v\}$}
                    \State Set $P_j' = P$ with $\{v\}$ added to $W_j^i$
                    \For{$P_j'$, $j\in [3]$}
                        \If{ $P_j'$ is a legal partition}
                            \State $A_{x_i}[P_j'] = A_{x_i}[P]$
                        \Else
                            \State $A_{x_i}[P'] = \bot$
                        \EndIf
                    \EndFor
                    \State $s(A_{x_i}[P']) = s(A_{x_i}[P]) + 1$, for $P' = (W_1^i, W_2^i, W_3^i, X^i \cup \{v\})$
                \EndFor
            \EndFor
        \EndFor
    \end{algorithmic}
\end{algorithm}

After all editable nodes $x$ are edited, Algorithm~\ref{algUpdTbl} updates the tables at each copy $x_i$, while avoiding any expensive recomputations caused by the non-editable children. It does so by utilizing the table at $x$ after subtracting all editable children, i.e., when $A_x$ only has the contribution of all non-editable children of $x$. Let $B_x = B_{x_i} \cup Q$, where $Q = \bigcup_{j \in a(x), j \neq i} W_{j}^x$. For any partition $P' = (W_1', W_2', W_3', X')$ of $B_{x_i}$, let $P$ be the partition $(W_1', W_2', W_3', X' \cup Q)$ of $B_x$. By setting $s(A_{x_i}[P']) = s(A_x[P]) - |Q|$ we retrieve the table at $x_i$ in time proportional to the table size times $k$. This is significantly faster than having to add an unbounded number of non-editable nodes to $x_i$, and is done at most three times for each editable node $x$. All the editable children of $x_i$ are then added to $x_i$, which is an expensive operation covered by the $\alpha$ potential of $\mathcal{T}$ (Thm~\ref{thm-split-covered}). Finally, separator vertices present in $V_{\mathcal{T}_{x_i}}$ are added to $x_i$.

\begin{algorithm}[hbt!]
    \caption{Split($r$, $\mathcal{T}$)}\label{algSplit}
    \begin{algorithmic}
        \State Let $S = \emptyset$. Begin a tree walk of $\mathcal{T}$ at root $r$
        \While{at an editable $x\in V_{\mathcal{T}}$}
            \If{$x = r$, i.e., $x$ is the root node}
                \State For every $i \in a(x)$, create a tree $\mathcal{T}_i$ with a single node $x_i$ such that $B_{x_i} = (C_i \cap B_x) \cup S$.
                \State Mark $x$ as pre-visited, and continue.
            \EndIf
            \If{$x$ is unvisited}
                \State Pick a good partition $P^x = (W_1^x, W_2^x, W_3^x, X^x)$ of $B_x$ w.r.t $p(x)$ if $x \neq r$.
                \State Update $S = S \cup X^x$
                \If{$x$ is editable}
                    \State Create copies $x_i$ of $x$ s.t. $B_{x_i} = W_i^x \cup X^x$ for all $i \in a(x)$ 
                    \State Link $x_i$ to $p(x)_i$ in $\mathcal{T}_i$, if $x \neq r$
                \Else
                    \State For $a(x) = \{i\}$, move $x$ and $\mathcal{T}^x$ to $\mathcal{T}_i$, by linking it to $p(x)_i$. 
                    \State(Note that $\mathcal{T}_i$ is the copy of $\mathcal{T}$ that contains vertices of $C_i$)
                    \State Mark $x$ as post-visited and return
                \EndIf
            \ElsIf{$x$ is pre-visited}
                \For{$i \in a(x)$}
                    \State UpdateTables($x$, $x_i$, $P^x$ $\mathcal{T}$, $\mathcal{T}_i$, S)
                \EndFor
            \EndIf
        \EndWhile
        \State Create new node $r'$ such that $B_{r'} = \bigcup_{x \in V_{\mathcal{T}}} X^x$. For $i \in a(r)$ add $r_i$ as grandchildren of $r'$, and add corresponding intersection nodes
        \State Begin a new tree traversal of the new tree $\mathcal{T}'$ from $r'$
        \While{at a node $x \in V_{\mathcal{T}'}$ that is not post-visited} 
            \State Merge $x$ into $p(x)$ if it violates the unique home property
        \EndWhile
    \end{algorithmic}
\end{algorithm}

\begin{lemma}\label{lem-editableIntersectionParent}
    An editable intersection node can only have editable main node children.
\end{lemma}
\begin{proof}
    Let $x$ be an editable intersection node. Let $y$ and $z$ be main nodes such that $x = p(y)$ and $z = p(x)$. By the definition of \gtd, if $x$ is the intersection node between main nodes $y$ and $z$, then it must hold that $B_z \cap B_y = B_x$. Since $x$ is editable, $B_x$ has a non-empty intersection with at least two different components $C_i$ and $C_j$ (for $i \neq j$). Since $B_x \subseteq B_y$, $B_y$ also intersects $C_i$ and $C_j$, and therefore, $y$ is editable.
    Thus, all main node children of an editable intersection node must be editable. 
\end{proof}

\begin{lemma}\label{lem-NonEditableRoot}
    On any path from the root to a leaf in $\mathcal{T}$, after a split has been chosen, the first non-editable node we encounter is an intersection node.
\end{lemma}
\begin{proof}
    This is an immediate consequence of Lemma~\ref{lem-editableIntersectionParent}.
\end{proof}

Lemma~\ref{lem-editableIntersectionParent} and Lemma~\ref{lem-NonEditableRoot} both imply that one need only search through the \textit{intersection node} children of a main node to perform the split. There is never any need to handle the potentially unbounded number of main node children of a non-editable intersection node. Conversely, we must go through all the main node children of an editable intersection node since they are all editable. However, the $\alpha$ potential function pays for this operation (Thm~\ref{thm-split-covered}).

\begin{theorem}\label{thm-split-covered}
    The $\alpha$ and $\delta$ potential of a \gtd $\mathcal{T}$ bounds the number of table updates performed during Split operations in a round of Algorithm~\ref{alg1}. Furthermore, the $\alpha-$potential bounds the increase in the $\beta$ and $\gamma$-potential during the Split.
\end{theorem}
\begin{proof}
    \noindent
    Consider a grouped tree decomposition $\mathcal{T}$. First we show that $\alpha(\mathcal{T})$ before a Split operation is strictly greater than $\alpha(\mathcal{T'})$, where $\mathcal{T'}$ is the tree decomposition obtained after all editable nodes in $\mathcal{T}$ are edited. Let $x$ be a non-root editable main node in $\mathcal{T}$ that is home to some vertex $\{v\}$ and is edited into $x_i$ for all $i \in a(x)$. The $\alpha$ potential of $x$ is $c_{\alpha} \, |B_x|$. Let $\{v\}$ be assigned to some component $W_j^x$ for $j\in a(x)$. Then, $x_j$ is the only copy of $x$ that has $\alpha$ potential, since for all $i \neq j, \, i \in a(x)$, $x_i$ is not the home of any vertex. By Def~\ref{def:split}, $|B_{x_j}| < |B_x|$, i.e., the $\alpha$ potential of $x_j$ is strictly less than that of $x$, for all non-root main nodes $x$ in $\mathcal{T}$. Now let $\{v\}$ be assigned to $X^x$. At the end of a Split, all separator vertices have their home in the root node. Thus no $x_i$ for $i \in a(x)$ has any $\alpha$ potential.     
    
    Let $r$ be the root of $\mathcal{T}$, with maximum size bag $B_r$, that is home to $|B_r|$ vertices. $r$ is split into $r_i$ for all $i\in a(r)$. If the unique home property was preserved at the root as well, there would exist a series of forget nodes above $r$ until the root of $\mathcal{T}$ had a bag of size one. As seen above, the $\alpha$ potential of these forget nodes and $r$ would also decrease in their copies in $\mathcal{T}'$. Since the unique home property is not preserved at the root, we give $r$ enough $\alpha$ potential to cover the $|B_r|$ forget nodes that would be required to preserve the property. This is covered by an appropriate factor in $f(x)$. Thus the $\alpha$ potential of $\mathcal{T}$ is strictly greater than the $\alpha$ potential of $\mathcal{T}'$. This decrease (multiplied by the constant factor in $c_{\alpha}$) pays for updating the tables of all the copies of an editable node and their intersection node parents.

    Each editable main node in $\mathcal{T}$ creates up to $6$ nodes in $\mathcal{T}'$ when it is edited -- i.e., its copies and their intersection node parents. These added nodes increase both the $\beta$ and $\gamma$ potential. The increase in $\gamma$ potential is $6\,c_{\gamma}$, while the increase in $\beta$ potential is $12\,c_{\beta}$ -- to account for the DFS-Status of the unvisited nodes. Since this is always a constant increase, it can be accounted for in $c_{\alpha}$. Thus the decrease in the $\alpha$ potential, covers the increase in the $\beta$ and $\gamma$ potentials during a round as well. The last step in a Split is to create the new root $r'$ with $B_{r'} = S$. Each copy $r_i$ of the original root $r$ is home to $|B_{r_i}|$ vertices. To maintain the unique home property, if $|B_{r_i} \setminus S| > 1$, we create a chain of forget nodes and their intersection nodes between $r_i$ and $r'$, each of which need to be given $\alpha$ potential. The bags of these nodes are decreasing in size, with the largest being of size $O(k)$, and we have at most $O(k)$ such nodes. Furthermore, the new root $r'$ would need enough potential to be home to $|S|$ vertices. The $\delta$ potential with $c_{\delta} k^2$ given to all maximum size bags, including $r$, covers this cost.
    Thus, the $\alpha$ and $\delta$ potentials bound the work done during a Split operation, and the merge performed at the end of a Split, to maintain the unique home property, is paid for by the $\gamma$ potential.
\end{proof}

\subsection{Analysis}\label{sec:analysis}
\begin{theorem}\label{thm-phi-accts-num}
    In each round of the algorithm the Potential Function $\phi(\mathcal{T})$ bounds the number of DP-table re-computations performed, and $\mathcal{O}^*(4^w \phi(\mathcal{T}))$ bounds the number of steps, where $w$ is the width of $\mathcal{T}$.
\end{theorem}

\begin{proof}    
    Thm~\ref{thm-split-covered} proves that the number of DP-table updates performed during all Split operations in a round is bounded by $\alpha(\mathcal{T}) + \delta(\mathcal{T})$. From Lemma~\ref{lem-move-nochange}, we know that the $\alpha$-potential does not change during a Move operation. Lemma~\ref{lem:betadecrease} and Lemma~\ref{lem:boundRotations} show that the number of rotation operations performed on a tree $\mathcal{T}$ can be counted using the $\beta$-potential. Furthermore, Thm~\ref{thm-split-covered} proves that any increase in the $\beta$-potential is paid for by the $\alpha$ and $\delta$ potentials. A Merge operation reduces the number of nodes in $\mathcal{T}$. Lemma~\ref{lem-merge-dec} shows that $\gamma(\mathcal{T})$ bounds the number of merge operations that may be performed per round. Thm~\ref{thm-split-covered} proves that any increase in the $\gamma$-potential is covered by the $\alpha$ and $\delta$ potentials.

    Lastly, since we only initialize a Split at a node with a maximum size bag, the number of SPLIT operations initiated in a round is at most the number of nodes with maximum size bags in $\mathcal{T}$. This is because a single Split may decrease the size of multiple maximum size bags. Furthermore, splitting a maximum size bag also requires the creation of a new root node whose bag contains the $\mathcal{O}(k)$ separator vertices.
    
    The $c_{\delta} k^2$ factor -- required to give $\mathcal{O}(k)$ potential for each of the additional $\mathcal{O}(k)$ vertices whose home is the new root -- accounts for the potential given to the new root node and the chain of forget nodes from the $r_i$'s to the new root, and it accounts for the work needed to construct these nodes. This gives us the last term of the potential function $\delta(\mathcal{T}) = c_{\delta} k^2 |V_{\mathcal{T}}^{max}|$. Thus, $\phi(\mathcal{T})$ bounds the number of DP-table updates done in each round of the algorithm. 
\end{proof}

\begin{remark}
    \textbf{About the constants:} We briefly describe the interactions between the constants $c_{\alpha}$, $c_{\beta}$, $c_{\gamma}$, and $c_{\delta}$. Since the $\alpha$ and $\delta$ potentials must cover the increase in the $\beta$ and $\gamma$ potentials after a split, $c_{\alpha}$ must be large enough to accommodate the added $c_{\beta}$, $c_{\gamma}$, and $c_{\delta}$, i.e., $c_{\alpha} > \max \{ c_{\beta}, c_{\gamma}, c_{\delta}$ \}. Similarly, $c_{\delta}$ must be large enough to give potential to the new root node, the copies of the old root node, and all the intersection and forget nodes added at the root to maintain the \gtd structure. $c_{\beta}$ must account for the updates made to the intersection nodes during a move, in addition to the main nodes, and $c_{\gamma}$ must be large enough to also pay for handling the intersection node parent of the main node being merged.
 \end{remark}

The algorithm begins with a tree decomposition $\mathcal{T}$ that has a $4$-approximation of the treewidth $k$. Let the size of the maximum size bag in $\mathcal{T}$ be $b$, such that $b \leq 4k+4$. 
When $b \in [3k+4, 4k+4]$, a $4$-partition split, i.e., $(C_1, C_2, C_3, S)$, creates DP-tables of size $\mathcal{O}(4^{4k}) = \mathcal{O}(256^k)$. However, note that it suffices to simply create a $3$-partition split \cite{Korhonen2021}, i.e., $(C_1, C_2, S)$, and let each $W_i^r$ to be of size at most $\frac{2}{3}\, |B_r|$ in any split, which needs DP-tables of size $\mathcal{O}(3^{4k}) = \mathcal{O}(81^k)$.
When $b \in [2k+3, 3k+3]$, creating a $4$-partition split, i.e., $(C_1, C_2, C_3, S)$, creates DP-tables of size $\mathcal{O}(4^{3k}) = \mathcal{O}(64^k)$. Since this is smaller than the time when $b \in [3k+4, 4k+4]$, our running time is determined by the former case.


\begin{theorem}\label{thm:runningTime}
    Algorithm~\ref{alg1} has a running time of $\mathcal{O}^*(81^k)n$.
\end{theorem}
\begin{proof}
    The running time of the algorithm is bounded by the most expensive operations performed, i.e., recomputations of DP-tables. $\phi(\mathcal{T})$ bounds the number of such DP-table recomputations that occur in a single round of the algorithm [Thm~\ref{thm-phi-accts-num}]. Furthermore $\phi(\mathcal{T}) = \mathcal{O}(k^2 \, n)$.
    When the width of $\mathcal{T}$, is in the range $[2k+2, 3k+2]$, the size of the DP-table is $\mathcal{O}(4^{3k})$. 
    Updating all entries in the table does not incur a further cost.
    Correspondingly, the amount of work done per round is $\mathcal{O}( k^2\, 4^{3k}\, n) = \mathcal{O}(64^k\,k^2\, n)$, and the algorithm runs for at most $\mathcal{O}(k)$ rounds. However the bag sizes are smaller in each subsequent round -- giving us a geometrically decreasing potential function that can be bound by the very first term in the series. This gives us a running time of $\mathcal{O}(k^2 \, 64^k\, n)$, which is $\mathcal{O}^*(64^k)n$. When $w \in [3k+3, 4k+3]$, the size of the DP-table should be $3^{4k}$, which gives us a running time of $\mathcal{O}(k^2\, 81^k\, n)$, which is $\mathcal{O}^*(81^k)n$. Thus, Algorithm~\ref{alg1} has a running time of $\mathcal{O}^*(81^k)n$.
\end{proof}

\section{Conclusion}

There have been quite a few FPT algorithms to produce tree decompositions, with running time linear in $n$ \cite{Korhonen2021, bodlaender1996linear, korhonen2022improved, bodlaender2016c}, and the focus has been shifting towards improving both the dependency on $k$ and the approximation ratio. 
In this paper, we present an $\mathcal{O}^*(81^k)n$ time 2-approximation of treewidth, which is a significant improvement on Korhonen's $\mathcal{O}^*(1728^k)n$-time 2-approximation algorithm~\cite{Korhonen2021}.
Unlike \cite{Korhonen2021}, our running time bottleneck occurs when the width is in the range $[3k+3, 4k+3]$, i.e. when we create a $4$-partition. It would be interesting to see if that case could be further optimized. In light of the recent work by Korhonen and Lokshtanov \cite{korhonen2022improved}, one could look into whether the $(1+\epsilon)$ approximation can be achieved in time $\mathcal{O}(c^k n)$ for some constant $c$. 

\newpage

\bibliography{references}
\newpage
\appendix

\end{document}